%% file: paper.tex
\documentclass[sigconf]{acmart}
\usepackage[english]{babel}

\usepackage{graphics}
\usepackage{algorithm}
\usepackage{dsfont}
\usepackage{algorithmic}
\usepackage{amssymb}
\usepackage{amsmath}
\usepackage{amsfonts}
\usepackage{epsfig}
\usepackage{url}
\usepackage{multirow}
\usepackage{xspace}
\usepackage{epstopdf}
\usepackage{graphicx}
\usepackage{color}
\usepackage{caption}
\usepackage{subcaption}
\usepackage{mdwlist}
\usepackage{tablefootnote}
\usepackage{tabularx}
\usepackage{array}
\usepackage{enumitem}
\newcolumntype{?}{!{\vrule width 2 pt}}

\input{dfn}

\newlength{\mfig}   
\newlength{\smfig}   
\newcommand{\alg}{{\bf {\tt Sub2Vec}}\xspace}
\newcommand{\randomwalk}{{\bf {\tt RandomWalk}}\xspace}

\newcommand{\nodevec}{{\bf {\tt Node2Vec}}\xspace}

\newcommand{\newman}{{\bf {\tt Newman}}\xspace}
\newcommand{\dlouvian}{{\bf {\tt Louvian}}\xspace}

\newcommand{\workplace}{{\tt WorkPlace}\xspace}
\newcommand{\highschool}{{\tt HighSchool}\xspace}

\newcommand{\dblp}{{\tt DBLP}\xspace}
\newcommand{\youtube}{{\tt Youtube}\xspace}

\newcommand{\polblogs}{{\tt PolBlogs}\xspace}
\newcommand{\texas}{{\tt Texas}\xspace}
\newcommand{\wisconsin}{{\tt Wisconsin}\xspace}
\newcommand{\washington}{{\tt Washington}\xspace}
\newcommand{\cornell}{{\tt Cornell}\xspace}
\newcommand{\facebook}{{\tt Facebook}\xspace}
\newcommand{\arxiv}{{\tt Astro-PH}\xspace}
\newcommand{\memetracker}{{\tt MemeTracker}\xspace}

 \newtheorem{informaldefinition}{Informal Definition}

\newcommand{\s}{\mathcal{S}}
\newcommand{\N}{\mathbf{M}}
\newcommand{\M}{\mathbf{M}}
\newcommand{\Sv}{\mathbf{S}}

\begin{document}

\title{Distributed Representation of Subgraphs}
\author{Bijaya Adhikari, Yao Zhang, Naren Ramakrishnan and B. Aditya Prakash}
\affiliation{\institution{Department of Computer Science, Virginia Tech}}
\email{ Email: [bijaya, yaozhang, naren, badityap]@cs.vt.edu}

\renewcommand{\shortauthors}{B. Adhikari et al.}
\begin{abstract}
Network embeddings have become very popular in learning
effective feature representations of networks. Motivated by the recent successes of 
embeddings in natural language processing, researchers have tried to find network embeddings in order to exploit machine learning algorithms for mining tasks like node classification and edge prediction. However, most of the work focuses on finding distributed representations of nodes, which
are inherently ill-suited to tasks such as community detection which are intuitively dependent on subgraphs. 

Here, we propose \alg, an unsupervised scalable algorithm to learn feature representations of arbitrary subgraphs. We provide means to characterize similarties between subgraphs and  provide theoretical analysis of \alg and demonstrate that it preserves the so-called local proximity. We also highlight the usability of \alg by leveraging it for network mining tasks, like community detection. We show that \alg gets significant gains over state-of-the-art methods and node-embedding methods. In particular, \alg offers an approach to generate a richer vocabulary of features of subgraphs to support representation and reasoning.

\end{abstract}

\date{}
\maketitle

\section{Introduction}
\label{sec:intro}

Graphs are a natural abstraction for representing relational data from multiple domains such as social networks,  protein-protein interactions networks, the World Wide Web, and so on. Analysis of such networks include classification~\cite{bhagat2011node}, link prediction~\cite{liben2007link}, detecting communities~\cite{girvan2002community, blondel2008fast}, and so on. Many of these tasks can be solved using machine learning algorithms. Unfortunately, since most machine learning algorithms require data to be represented as features, applying them to graphs is challenging due to their high dimensionality and structure. In this context, learning meaningful feature representation of graphs can help to leverage existing machine learning algorithms more widely on graph data.

Apart from classical dimensionality reduction techniques (see related work), recent works~\cite{perozzi2014deepwalk, grover2016node2vec, wang2016structural, tang2015line} have explored various ways of learning feature representation of nodes in networks exploiting relationships to vector representations in NLP (like word2vec~\cite{mikolov2013distributed}). However, application of such methods are limited to binary and muti-class node classification and edge-prediction. It is not clear how one can exploit these methods for other tasks like community detection which are inherently based on subgraphs and node embeddings result in loss of information of the subgraph structure. Embedding of subgraphs or neighborhoods themselves seem to be better suited for these tasks. Surprisingly, learning feature representation of networks themselves (subgraphs and graphs) has not gained much attention thus far. In this paper, we address this gap by  studying the problem of learning distributed representation of subgraphs.  
Our contributions are:
\begin{enumerate}
\item We propose \alg, a scalable subgraph embedding method to learn features for arbitrary subgraphs that maintains the so-called local proximity.

\item We also provide theoretical justification of network embedding using \alg, based on language modeling tools. We also propose meaningful ways to measure how similar two subgraphs are to each other.

\item We conduct multiple experiments over large diverse real datasets to show correctness, scalability, and utility of features learnt by \alg in several tasks. In particular we get upto 4x better results in tasks such as community detection compared to just node-embeddings. 

\end{enumerate}

\begin{figure*}[htb]
\begin{center}

\begin{tabular}{ccc|c}
		\includegraphics[width=1.2in]{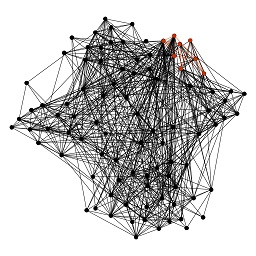} &
       \includegraphics[width=1.2in]{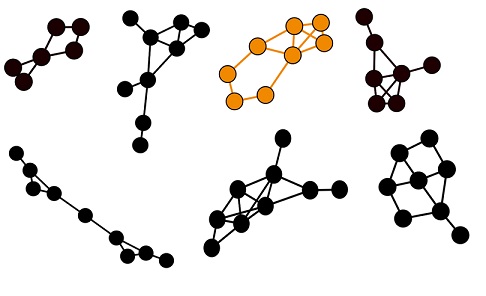} &
       \includegraphics[width=1.2in]{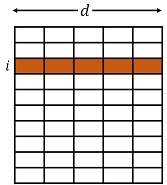} &
       \includegraphics[width=1.2in]{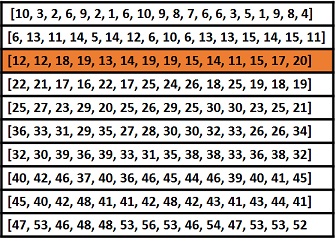} \\
      (a) A network $G$   & (b) A set, $\s$, of subgraphs of $G$  & (c) embedding learned for each subgraph  & (d)  Intermediate neighborhoods\\
      & & &  on each subgraph
    \end{tabular}
\end{center}
\caption{\textbf{An overview of our \alg. Our input  is a set of subgraphs $\s$ drawn from a network $G$. We obtain $d$ dimensional embedding of subgraphs such that we maximize the likelihood of observing intermediate neighborhoods. }}
\label{fig:intro}
\end{figure*}
 
 The rest of the paper is organized as follows: we first formulate and motivate our problem, then present \alg, discuss experiments, and finally present related work, discussion and conclusions.

\section{Problem Formulation}
\label{sec:prob}

In this paper, we are interested in embedding \emph{subgraphs} into a low dimensional continuous vector space. As shown later, the vector representation of subgraphs enables us to apply off-the-shelf machine learning algorithms directly to solve subgraph mining tasks. For example, to group subgraphs together, we can apply clustering algorithms like KMeans directly. Figure~\ref{fig:intro} (a-c) gives an illustration. Given a set of subgraphs (Figure~\ref{fig:intro} (b)) of a graph $G$ (Figure~\ref{fig:intro} (a)), we learn  a low-dimensional feature representation of each subgraph (Figure~\ref{fig:intro}(d)).

Now we are ready to formulate  our Subgraph Embedding problem. We are given a graph $G(V,E)$ where $V$ is the vertex set, and $E$ is the associated edge-set (we assume undirected graphs here, but our framework can be easily extended to directed graphs as well). We define $g_i(v_i, e_i)$ as a subgraph of $G$, where $v_i \subseteq V$ and $e_i \subseteq E$. For simplicity, we write $g_i(v_i, e_i)$ as $g_i$. As input we require a set of subgraphs  $\s = \{g_1,g_2, \dots, g_n \}$. Our goal is to embed  subgraphs in $\s$ into $d$-dimensional feature space $\mathbb{R}^d$, where $d<< |V|$. In addition, we want to ensure the subgraph proximity is well-preserved in such a $d$-dimensional space.  
In this paper, we consider to preserve the ``local neighborhood" of each subgraph $g_i$. The idea is that if two subgraphs share common structure, then their vector representations in $\mathbb{R}^d$ are close. We call such a measure \emph{Local Proximity}.  

\begin{informaldefinition}\textbf{(Local Proximity).} Given two subgraphs $g_i(v_i, e_1)$ and $g_j(v_j, e_j)$, the local proximity between $g_i$ and $g_j$  is larger if the commonly induced subgraph is larger. 
\end{informaldefinition}

Intuitively, local proximity measures how many nodes, edges, and paths are shared by two subgraphs. For illustration of the local proximity, let us consider an example. In Figure \ref{fig:toy}, suppose $g_1$, $g_2$, and $g_3$ are subgraphs induced  by nodes $\{a, b, c, e \}$ and $\{b, c, d, e\}$, and $\{d, e, f, j\}$. Since, the subgraph commonly induced by $g_1$ and $g_2$ is larger than the subgraph commonly induced by $g_1$ and $g_3$, we say $g_1$ and $g_2$ to be more ``locally proximal'' to each other than $g_1$ and $g_3$. Note that the local proximity is not just the Jaccard similarity of nodes in the two subgraphs, as it also takes the connections among the common nodes into account.

\begin{figure}[htb]
	\includegraphics[width=0.3\textwidth]{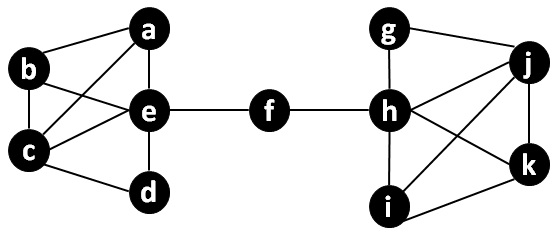}
	\caption{A toy network}
	\label{fig:toy}
\end{figure}

Having defined the local proximity of two subgraphs, we focus on learning vector representations of subgraphs such that the local proximity is preserved. Formally, our Subgraph Embedding problem is, 

\begin{problem} 
	\label{prob:embedding}
	Given a graph $G(V,E)$, $d$ and set of $\s$ subgraphs (of $G$) $\s = \{g_1,g_2, \dots, g_n \}$, learn an embedding function $f: g_i \rightarrow \mathbf{y_i} \in \mathbb{R}^d$ such that $\textbf{Local Proximity}$ among subgraphs is preserved.
\end{problem}

According to Problem~\ref{prob:embedding}, if $g_i$ and $g_j$ are closer to each other in terms of the local proximity that $g_k$ and $g_k$ then the $ \mathbf{sim} \left( f(g_i),f(g_j) \right)$ has to be greater than $ \mathbf{sim} \left( f(g_i),f(g_k) \right)$, where $\mathbf{sim}(\mathbf{x}, \mathbf{y})$ is a similarity metric between two real vectors $x$ and $y$ in $\mathbb{R}^d$. Hence, if we embed the subgraphs in Figure \ref{fig:toy} from the previous example, then a correct algorithm to solve Problem \ref{prob:embedding} has to ensure that $\mathbf{sim} \left( f(g_1),f(g_2) \right) > \mathbf{sim} \left( f(g_i),f(g_j) \right)$. We propose an efficient algorithm  for Problem \ref{prob:embedding} based on two different optimization objectives in the next section.

A natural question to ask is that if there are other metrics of subgraph similarity. Indeed, one can think of other measures of proximity, which may result in different embeddings. We will discuss this point further in Section~\ref{sec:conclusions}.

\section{Learning Feature Representations}
\label{sec:proposed}
In this section, we propose two optimization objectives for Problem~\ref{prob:embedding} and propose an unsupervised deep learning technique to optimize the objectives.
 
Mikolov et al. proposed the continuous bag of words and skip-gram models in~\cite{mikolov2013distributed}, which have been extensively used in learning continuous feature representation of words. Building on these two models, Le et al.~\cite{le2014distributed} proposed two models: the Distributed Memory of Paragraph Vector (PV-DM), and the Distributed Bag of Words version of Paragraph Vector (PV-DBOW), which can learn continuous feature representations of paragraphs and documents.

Our main idea is to pose our feature learning problem as a maximum likelihood problem by extending PV-DM and PV-DBOW  to  networks.  The direct analog is to treat each node as a word, and each subgraph as a paragraph. The edges within a subgraph can be thought as the adjacency relation of two words in a paragraph.  
PV-DBOW and PV-DM assume  that if two paragraphs share similar sequence of words, they are close in the embedded feature space. The local proximity of subgraphs naturally follows the above assumption. 
Hence, we can leverage deep learning techniques in~\cite{le2014distributed} for our subgraph embedding problem. PV-DBOW and PV-DM learn a latent representation by maximizing a distribution of word co-occurrences (using either n-gram or skip-gram model). Similarly, in this paper, we maximize a distribution of ``node neighborhood". The so-called  ``node neighborhood" is generated by subgraph-truncated random walks (see details in Section~\ref{sec:nei}). We call our models  \emph{Distributed Bag of Nodes version of Subgraph Vector }(\alg-DBON) and \emph{Distributed Memory version of Subgraph Vector} (\alg-DM) respectively. 

Next, we will introduce \alg-DM, \alg-DBON first, then study how to generate ``node neighborhood" and give a justification from matrix multiplication view. Finally, we summarize our algorithm \alg.

\subsection{\alg-DM}
\label{sub:svdm}

In the \alg-DM model, we seek to predict a node $u$ given other nodes in $u$'s neighborhoods and the subgraph $u$ belongs to.  Consider the subgraph $g_1$ (a subgraph induced by nodes $\{a, b, c, e\}$)  in Figure \ref{fig:toy}. Suppose the sequence of nodes returned by random walks in $g_1$ is $[ a, b, c]$, and we consider neighborhood of distance 2, then the model asks to predict node $c$ given subgraph $g_1$, and its predecessors ($a$ and $b$), i.e., $\Pr(c | g_1, \{a, b\})$.
 
 More precisely, given a $G'(V', E')$ as the union graph of all the subgraphs in $\s = \{g_1, g_2, \dots , g_n\}$, where $V'=\bigcup_i {v_i}$ and $E'=\bigcup_i {e_i}$, consider a function $m$: $V'\rightarrow \mathbb{R}^d$ ($m(n) = \mathbf{x} $). 
 We define $\N$ as a $d\times |V'|$ node vector matrix, where each column is $m(n)$ (the vector representation of nodes $n \in V'$). 
  Similarly, we define function $f(g_i)$ as the embedding function for subgraph $g_i$, where $f(g_i)$ is a $d$-dimensional vector. We denote $\Sv$ as the  subgraph matrix, where each column is $f(g_i)$ for all subgraphs in $\s$.
  The matrices $\N$ and $\Sv$ are indexed by node and subgraph ids. In \alg-DM, we use the node and subgraph vectors to predict the next node in the neighborhood $N_n$. We assume $N_n$ is given, and will discuss  $N_n$ later in Section~\ref{sec:nei}.  

Now, given a node $n$ and its neighborhood $N_n$ and the subgraph $g_i$ from which the $N_n$ is drawn, the objective of \alg-DM is to maximize the following:

\begin{equation}
\max_{f}  \sum_{g_i \in \s} \sum_{n \in g_i} \log( \Pr(n | m(N_n), f(g_i)),
\end{equation}
where $\Pr(n | m(N_n), f(g_i))$ is the probability of predicting node $n$ in $g_i$ given the vector representations of its neighborhood $m(N_n)$ and the subgraph from which the node and its neighborhood is drawn, $f(g_i)$.  Note that for ease of description, we extend the function $m$ from a node to a node set (neighborhood $N_n$). $\Pr(n | m(N_n), f(g_i))$  is defined using the softmax function:

\begin{equation}
\label{svdbowprob}
\Pr(n | m(N_n), f(g_i) )= \frac{e^{\mathbf{U}_n \cdot h(m(N_n),f(g_i))}}{\sum_{v \in V} e^{\mathbf{U}_v \cdot h(m(N_n),f(g_i))}}
\end{equation}
where matrix $\mathbf{U}$ is a softmax parameter and $h(\mathbf{x},\mathbf{y})$ is average or concatanation of vectors $\mathbf{x}$ and $\mathbf{y}$~\cite{le2014distributed}. In practice, to compute Equation~\ref{svdbowprob}, hierarchical softmax is used~\cite{mikolov2013distributed}.

\subsection{\alg-DBON}

In the \alg-DBON model,  we want to predict the nodes in the subgraph given only the subgraph vector $f(g_i)$. 
For example, consider the same example in Section~\ref{sub:svdm}: the subgraph $g_1$  in Figure \ref{fig:toy}, and the node sequence $[a,b,c]$ generated by random walks. Now, in the \alg-DBON model the goal is to predict the neighborhood $\{a, b, c\}$ given the subgraph $g_1$. This model is parallel to the popular skip-gram model.

Formally, given a subgraph $g_i$, and neighborhood $N$ drawn from $g_i$, the objective of \alg-DBON is the following:
\begin{equation}
\max_{f}  \sum_{g_i \in \s} \sum_{N \in g_i} \log( \Pr(N | f(g_i)),
\end{equation}
where $\Pr(N | f(g_i)$ is also a softmax function, i.e., 
\begin{equation}
\label{softmax}
\Pr(N | f(g_i) = \frac{e^ {m(N). f(g_i)}}{\sum_{N \in G} e^{ m(N). f(g_i)}},
\end{equation}

Since computing Equation \ref{softmax} involves summation over all possible neighborhoods, we use negative sampling to optimize it.
The negative sampling objective is as follows:

{
\small
\begin{equation}
\label{obj:ns}
\mathcal{L} = \sum_{g_i \in \s} \sum_{c \in g_i} \#(g_i, c) \log( \sigma( g(c) \cdot f(g_i)) + k  \mathrm{E}_{c_N ~ P}[\log( \sigma( -  g(c_N) \cdot f(g_i))]
\end{equation}
}
where $k$ is a parameter for negative sampling, $c$ is a context generated by random walks, and $\sigma(x) = \frac{1}{1+ e^{-x}}$.

\subsection{Subgraph Truncated Random Walks}
\label{sec:nei}
Our problem seeks to preserve the local proximity between subgraph in $\s$. As mentioned in Section~\ref{sec:prob}, intuitively the local proximity measures how many nodes, edges, and paths are shared by two subgraphs. However, quantify local proximity is challenging.
A possible way to measure the local proximity between two subgraphs $g_i$ and $g_j$, would be to look at their neighborhoods, and compare every neighborhood in $g_i$ with every neighborhood in $g_j$. However, it is not feasible as we have a large number of neighborhoods. Another approach to measure local proximity is that we can enumerate all possible paths in each subgraphs. However, there are exponential number of paths in each subgraphs. To bypass these challenges, we resort to random walks to implement the local proximity. 

Given a set of subgraphs $\s = \{g_1,g_2, \dots, g_n \}$, we generate neighborhood in each $g_i \in \s$ by fixed length subgraph-truncated random walks. Specifically, for a subgraph $g_i$, we choose a node $v_1$ from nodes in $g_i$ uniformly at random. Next we generate a sequence of nodes  $v_1, v_2, v_3 \dots v_k$ to get a random walk of length $k$, where $v_j$ is a node chosen from the neighbors of node $v_{j-1}$ uniformly at random. We repeat the process for each subgraph in $\s$. Overlaps in the random walks of $g_i$ and $g_j$ serve as a metric for local proximity. The intuition is that if the subgraph commonly induced by $g_i$ and $g_j$ is large, then we have more overlaps in their random walks.

Apart from being tractable in capturing the notion of local proximity between subgraphs, random walks have other advantages. First, the notion of neighborhood in other data types, such as texts, is naturally defined due to the sequential nature of text data. However, graphs are not sequential, hence  it is more challenging to define the neighborhoods of subgraphs. Random walks help sequentialize subgraphs. Moreover, 
random walks generate meaningful sequences, for example, the frequency of nodes in random walk follows power law distribution \cite{perozzi2014deepwalk}.

\subsection{Matrix Multiplication based Justification of our Model}
Here we demonstrate that optimizing the objective function of SV-DBON with negative sampling preserves the local proximity of subgraphs. Leveraging the idea in~\cite{levy2014neural}, we can write Equation \ref{obj:ns}
as a factorization of matrix $\mathbf{M}$,  where  each element $M_{ij}$ corresponds to subgraph $i$ and context $j$:

\begin{equation}
M_{ij} = \log(\frac{\#(\text{context j in subgraph i})}{\#(\text{context j in D})}) + \log(\frac{|D| \cdot w}{k \cdot l}),
\end{equation}

$k$ is a negative sampling parameter, $w$ is a window size of context, and $l$ is a length of a random walk in each subgraph.  
Note that if  subgraph $i$ in $D$ has contexts $j$ that is never observed,  then in $\mathbf{M}$, $M_{ij} = \log(0) = - \infty$. A common practice in NLP is to replace $\N$ with $\M^{0}$ where, $\M^{0} = 0$ if  $\#(\text{context j in subgraph i}) = 0$. 

Suppose $\M^{0}_{a}$ is the a-th row in matrix $\M^{0}$, and $\M_a \cdot \M_b$ is a dot-product. Now, we have the following lemma.

\begin{lemma}
\label{lemma:factorization}
Assuming random walks in subgraphs $g_a$ and $g_b$ visit every path of size $w$ at least once, then
\begin{equation}
\M^{0}_{a} \cdot \M^{0}_{b} \geq x \log^2(\frac{|D|w}{|\s|kl}) ,
\end{equation}
where $S$ is set of input subgraphs in the data, $D$ is the set of all the subgraph-context pairs observed ,and $x$ is the number of overlapping paths of length $w$ in subgraphs $g_a$ and $g_b$.
\end{lemma}

\begin{proof}
Now, by the definition of dot product, we have the following:
\begin{equation}
\M^{0}_{a} \cdot \M^{0}_{b} = \sum_{j=1}^C \left[ \log \left(\frac{\#(\text{j, a})\cdot |D|\cdot w} {\#(\text{j, D}) \cdot k \cdot l} \right) \right] \left[ \log \left(\frac{\#(\text{ j, b})\cdot |D|\cdot w} {\#(\text{j, D}) \cdot k \cdot l} \right) \right],
\end{equation}
where $\#(\text{j, a})$ is the number of times context $j$ appears in subgraph $g_a$.

Now, we know that maximum value of $\#(\text{j, D})$ is ${N \cdot (l - w + 1 )}$ when random walk produces only context $j$. And the minimum value of $\#(\text{j, a})$ is $1$, as the random walk visits each path in the subgraph if it exists. Now, summing only over non-zero entries.

\begin{equation}
\M^{0}_{a} \cdot \M^{0}_{b} \geq \sum_{j \in {\#(\text(j,a)) \neq 0 , \#(\text(j,b)) \neq 0}} \left[ \log^2 \left(\frac{|D|\cdot w} { N \cdot k \cdot l (l-w+1)} \right) \right]
\end{equation}

Now using the fact that $l \geq (l -w +1)$ for any $w < l$ and that there are exactly $x$ non-zero entries in the summation, we get

\begin{equation}
\M^{0}_{a} \cdot \M^{0}_{b} \geq x  \log ^2 \left(\frac{|D|\cdot w} { N \cdot k \cdot l^2} \right)
\end{equation}
\end{proof}

Lemma \ref{lemma:factorization} shows that as the number of overlapping paths increases, the lower bound of any $\M^{0}_{a} \cdot \M^{0}_{b}$ (corresponding to subgraphs $g_a$ and $g_b$) increases as well. Since optimizing $\alg$'s objective is closely related to the factorization of matrix $\M^0$, we can expect the embedding of subgraphs with higher overlaps to be closer to each other in the feature space. Hence, \alg preserves the local proximity.

\subsection{Algorithm}

\begin{algorithm}
\caption{\alg}
\label{alg:alg}
\begin{algorithmic} [1]
\REQUIRE Graph $G$, subgraph set $\s = \{g_1,g_2, \dots, g_n \}$, length of the context window $w$, dimension $d$

\STATE walkSet = $\{ \}$ 
\FOR{\textbf{each} $g_i$ in $s$}
	\STATE walk = \randomwalk($g_i$)
	\STATE walkSet[$g_i$] = walk
\ENDFOR
\STATE f = StochasticGradientDescent(walkSet, $d$, $w$) 
\RETURN f
\end{algorithmic}
\end{algorithm}

\begin{algorithm}
\caption{\alg : StochasticGradientDescent(walkSet, $d$, $w$) }
\label{alg:alg2}
\begin{algorithmic} [1]
\STATE randomly intialize features $f$
\FOR{\textbf{each} walk $i$ in walkset}
\FOR{\textbf{each} randomly sampled Neighborhood $N$ in walk $i$}
 \STATE Compute $\mathcal{L}(f)$ based in SV-DM or SV-DBON objective
\STATE $f$ = $f - \eta \times \nabla \mathcal{L}(f$
\ENDFOR
\ENDFOR
\end{algorithmic}
\end{algorithm}

In our algorithm, we first generate the neighborhood in each subgraph by running random walk. We then
learn the vector representation of the subgraphs based on the random walks generated on each subgraph. Then stochastic gradient descent is used to optimize SV-DBON/ SV-DM objectives. The complete pseudocode is presented in Algorithms  \ref{alg:alg} and \ref{alg:alg2}.

\section{Experiments}
\label{sec:experiments}
We briefly describe our set-up next. All experiments are conducted  using
a 4 Xeon E7-4850 CPU with 512GB 1066Mhz RAM. We set the length of the random walk as 1000 and following literature \cite{grover2016node2vec}, we set dimension of the embedding as 128 unless mentioned otherwise. The code was implemented in Python and we will release it for research purposes. 
We answer the following questions in our experiments: 
\begin{enumerate}[wide, labelwidth=!, labelindent=0pt, itemsep=0pt,nolistsep]
\item[Q1.] Are the embeddings learnt by \alg useful for community detection? 
\item[Q2.] Are the embeddings learnt by \alg effective for link prediction?
\item[Q3.] How scalable is \alg for large networks? 
\item[Q4.] Do parameter variations in \alg lead to overfitting? 
\item[Q5.] Are the representations learnt by \alg  meaningful?
\end{enumerate}

\par \noindent
\textbf{Datasets.} We run \alg on multiple real world datasets from multiple domains like social-interactions, co-authorship, social networks and so on of varying sizes. See Table~\ref{tab:datasets}.
\begin{enumerate}[wide, labelwidth=!, labelindent=0pt, itemsep=0pt,nolistsep]
\item  \workplace is a publicly available social contact network between employees of a company with five departments\footnote{\protect\label{note1}http://www.sociopatterns.org/}\!. Edges indicate that two people were in proximity of each other. 
 
 \item \highschool is a social contact network\textsuperscript{\normalfont\ref{note1}}\!. Nodes are high school students belonging to one of five different sections and edges indicate that two students were in vicinity of each other.  

\item \texas, \cornell, \washington, \wisconsin are networks from the WebKB dataset\footnote{http://linqs.cs.umd.edu/projects/projects/lbc/}\!. These are networks of webpages and hyperlinks.   

\item \polblogs is a directed network of hyperlinks between weblogs on US politics, recorded in 2005. 

\item \arxiv and \dblp are coauthorship networks from Arxiv High-energy Physics and DBLP bibliographies respectively, where two authors have an edge if they have co-authored a paper. 

\item \facebook \cite{leskovec2012learning} is an anonymized social network where nodes are Facebook users  and edges indicate that two users are friends.

\item \youtube is a social network, where edges indicate friendship between two users. 

\end{enumerate}

{
\small
\begin{table}

\centering
\caption{\label{tab:datasets}Datasets Information.}
\begin{tabular}{|c|c|c|c|}
	\hline
 \textbf{Dataset}  & \textbf{$|V|$} & \textbf{$|E|$}  & \textbf{Domain} \\ \hline

\workplace \cite{NWS:9950811}  & 92 & 757 & contact  \\ \hline
\cornell \cite{sen:aimag08} & 195 & 304 & web\\ \hline
\highschool \cite{10.1371/journal.pone.0107878} &  182 & 2221 & contact\\ \hline
\texas \cite{sen:aimag08}  & 187 & 328 & web\\ \hline
\washington \cite{sen:aimag08}   & 230 & 446  & web\\ \hline
\wisconsin  \cite{sen:aimag08} & 265 & 530 & web\\ \hline
\polblogs  \cite{adamic2005political}& 1490 & 16783 & web\\ \hline
\facebook \cite{leskovec2012learning}  & 4039 & 88234 & social-network\\ \hline
\arxiv \cite{leskovec2007graph}  & 18722 & 199110 & co-author\\ \hline
\dblp \cite{yang2015defining} &  317k & 1.04 M & co-author \\ \hline
\youtube \cite{yang2015defining} &  1.13M & 2.97M & social\\ \hline

\end{tabular}

\end{table}
}

\subsection{Community Detection}

\par \noindent 
\textbf{Setup.} Here we show how to leverage \alg for the well-known community detection problem. A community of nodes in a network is a coherent group of nodes which are roughly densely connected among themselves and sparsely connected with the rest of the network. As nodes in a community are densely connected to each other, we expect neighboring nodes in the same community to have a similar surrounding. We know that \alg embeds subgraphs while preserving local proximity. Therefore, intuitively we can use features generated by \alg to detect communities. 

Specifically, we propose to solve the community detection problem using \alg by embedding the surrounding neighborhood of each node. First, we extract the neighborhood $C_v$ of each node $v \in V$ from the input graph $G(V,E)$. Then we run \alg on $\s = \{C_v | v \in V \}$ to learn feature representation of $f(C_v)$ for all $C_v \in \s$. We then use a simple clustering algorithm (K-Means) to cluster the feature vectors $f(C_v)$ of all ego-nets. Cluster membership of ego-nets determines the community membership of the ego.   The complete pseudocode is in Algorithm \ref{alg:com}.

\begin{algorithm}
\caption{Community Detection using \alg}
\label{alg:com}
\begin{algorithmic} [1]
\REQUIRE A network $G(V,E)$, \alg parameters, $k$ number of communities
\STATE neighborhoodSet = $\{ \}$ 
\FOR{\textbf{each} $v$ in $V$}
	\STATE neighborhoodSet = neighborhoodSet $\cup$ neighbordhood of $v$ in $G$.
\ENDFOR
\STATE vecs = \alg(neighborhoodSet, $w$, $d$)
\STATE clusters = K-Means(vecs, $k$)
\STATE return clusters
\end{algorithmic}
\end{algorithm}

In Algorithm \ref{alg:com},  we define neighborhood of each node to be its ego-network for dense networks (\highschool and \workplace) and 2-hop ego-networks for sparse networks. The ego-network of a node is the subgraph induced by the node and its neighbors.  Similarly,  the 2-hop ego-network of a node is defined as the subgraph induced by the node, its neighbors, and neighbors' neighbors. 

We compare \alg with various traditional community detection algorithms and network embedding based methods. \newman \cite{girvan2002community} is a community detection algorithm based on betweenness. It is a greedy agglomerative hierarchical clustering algorithm. \dlouvian \cite{blondel2008fast} is a greedy optimization method. \nodevec is a network embedding method which learns feature representation of nodes in the network which we then cluster to obtain communities.

We run \alg and baselines on the following networks with ground truth communities and compute Precision, Recall, and F-1 score to evaluate all the methods. 

\begin{enumerate}[wide, labelwidth=!, labelindent=0pt, itemsep=0pt]
\item  \workplace: Each department as a ground truth community. 
 
 \item \highschool: Each section as a ground truth community. 

\item \texas, \cornell, \washington: Each webpage belongs to one of five classes: course, faculty, student, project, and staff, which serve as ground-truth.

\item \polblogs: Conservative and liberal blogs as ground-truth communities. 

\end{enumerate}

{\footnotesize
\begin{table*}[htb]
\centering
\caption{\alg easily out-performs all baselines in all datasets. Precision \textbf{P}, Recall \textbf{R}, and \textbf{F-1} score, of various algorithms for community detection. Winners in F-1 score have been bolded for each dataset.}
\label{tab:com}
\begin{tabular}{|c ? c|c|c ? c|c|c ?c|c|c ? c|c|c ? c|c|c ? c|c|c ? c|c|c ?}
\hline
 & 
\multicolumn{3}{ c ?}{\workplace} &\multicolumn{3}{ c ? }{\highschool} &\multicolumn{3}{ c?}{\polblogs} & \multicolumn{3}{ c?}{\texas} & \multicolumn{3}{c?}{\cornell} & 
\multicolumn{3}{c?}{\washington} & \multicolumn{3}{c?}{\wisconsin} \\ \hline
\textbf{Method} &
 \textbf{P} & \textbf{R} & \textbf{F-1} & \textbf{P} & \textbf{R} & \textbf{F-1} & \textbf{P} & \textbf{R} & \textbf{F-1} & \textbf{P} & \textbf{R} & \textbf{F-1} & \textbf{P} & \textbf{R} & \textbf{F-1} &
 \textbf{P} & \textbf{R} & \textbf{F-1} & \textbf{P} & \textbf{R} & \textbf{F-1} \\ \hline
  
  \newman &  0.26 & 0.27 & 0.27 & 0.23 & 0.32 & 0.27 &  0.67  & 0.64  & 0.66 & 0.43 & 0.15 & 0.22  & 0.38   & 0.25  & 0.30 &  0.32 & 0.87  & 0.47 & 0.35  & 0.13 & 0.19 \\\hline
  
  \dlouvian &  0.57 & 0.04 & 0.07 & 0.49 & 0.04 & 0.08 & 0.91  & 0.83 & 0.87 & 0.54 & 0.14 & 0.23 & 0.36  & 0.15  & 0.22 & 0.45  & 0.1 & 0.16 & 0.40  & 0.12 & 0.19  \\ \hline
  
    \nodevec & 0.26 & 0.21 & 0.23 & 0.21 & 0.22 & 0.22 & 0.92  & 0.92 &  0.92 &  0.41 & 0.63 & 0.50 & 0.30  & 0.36 & 0.33  & 0.37 & 0.45  & 0.40  & 0.34 & 0.24  & 0.29  \\\hline
 
  \alg DM &  0.87 & 0.69 & \textbf{0.77} &  0.95 & 0.95 & \textbf{0.95} &  0.92 & 0.93 &  \textbf{0.93} & 0.49  & 0.57 & \textbf{0.53} & 0.34  & 0.47  & 0.39  & 0.45 & 0.64  & \textbf{0.53}  & 0.40 & 0.42  & \textbf{0.41}  \\\hline
 
  \alg DBON &  0.86 & 0.67 & 0.77 &  0.94 & 0.94  & 0.94 &   0.92 & 0.92 &  0.92 & 0.44  & 0.59 & 0.51 & 0.31 & 0.55  & \textbf{0.40}  & 0.43 & 0.66  & 0.52  & 0.35 & 0.41  & 0.38  \\\hline
\end{tabular}
\end{table*}
}

\par \noindent
\textbf{Results.} 
See Table \ref{tab:com}. Both versions of \alg significantly and consistently outperform all the baselines (\emph{upto a factor of \textbf{4} times against closest competitor, \nodevec}). We do better than \nodevec because intuitively, we learn the feature vector of the neighborhood of each node for the community detection task; while \nodevec just does random probes of the neighborhood. Precision for \dlouvian is high in dense networks as it outputs small communities and recall is consistently poor across all datasets for the same reason, while  for \newman the performance is not consistent. Performance of \nodevec is satisfactory in the sparse networks like \polblogs and \texas, but it is significantly worse for dense networks like \workplace and \highschool. On the other hand, performance of \alg is even more impressive in these networks. 

In Figure \ref{fig:comDec}, we plot the community structure of the \highschool dataset.  In the \highschool dataset, we consider five sections as the ground truth community. In the figure, the color of nodes indicate the community membership. The figure highlights the superiority of \alg compared to \nodevec. The communities discovered by \alg matches the ground truth very closely, while those discovered by \nodevec appear to be \emph{near random}.

\begin{figure*}[htb]
\begin{center}

\begin{tabularx}{\textwidth}{>{\centering\arraybackslash}X>{\centering\arraybackslash}X>{\centering\arraybackslash}X}
		\includegraphics[width=2.0in]{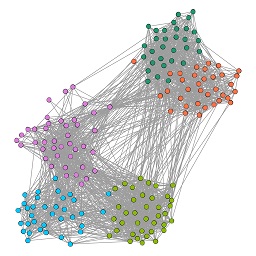} &
       \includegraphics[width=2.0in]{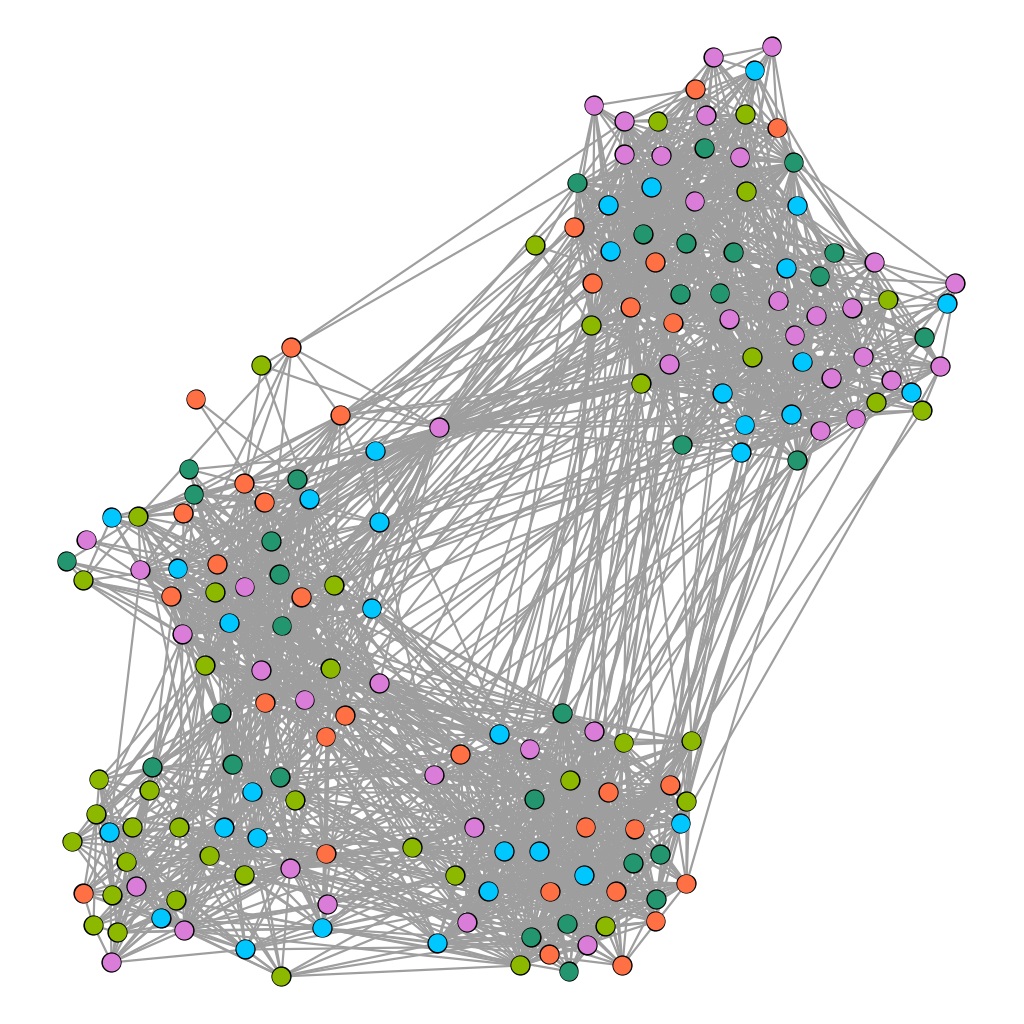} &
       \includegraphics[width=2.0in]{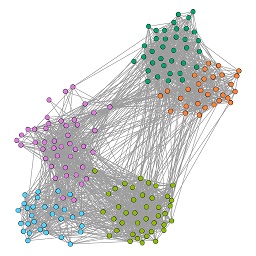} \\
      (a) Ground Truth   & (b) Result of \textsc{node2vec} & (c) Result of \alg
    \end{tabularx}
\end{center}
\caption{\textbf{Visualization of community detection in dense \highschool network. Communities obtained by clustering ego-nets vectors returned by \alg matches the ground truth, while the result from \nodevec appears to be random.}}
\label{fig:comDec}
\end{figure*}

\subsection{Link Prediction}

\par \noindent 
\textbf{Setup.}
In this section, we focus on the Link Prediction problem. Given a network $G(V,E)$, the link prediction problem asks to predict the likelihood of formation of an edge between two nodes $v_1 \in V$ and $v_2 \in V$, such that $(v_1,v_2) \notin E$. It is well known that nodes with common neighbors tend to form future links \cite{liben2007link}. For example, in a social network two individuals who have multiple friends in common have higher chances of eventually forming a friendship.  
It is evident from the example that likelihood of future edges depends on the similarity of neighborhood around each end-point. Hence we propose exploiting the embeddings of ego-nets of each node obtained from \alg to predict whether two nodes will form an edge.
 
Specifically, we first hide a $P$ percentage of edges randomly sampled from the network, while ensuring that the remaining network remains connected. We consider these ``hidden'' edges as the ground truth. Then we extract the ego-network, $C_v$, for each node $v \in V$. We then run \alg on $\s = \{C_v | v \in V \}$ and use the resulting embedding to predict link. Following methodology in literature~\cite{wang2016structural}, to evaluate our method, we calculate the Mean Average Precision (MAP). To calculate MAP first we compute Precision@K, as  $\text{Precision@k}(v)= \frac{\sum_{i < k} \mathds{1}(v, v_i)}{k}$. Here $v_i$ is the $i^{th}$ node predicted to have edge with node $v$ and  $\mathds{1}(v, v_i) = 1$ if $(v, v_1)$ is in the ground truth, $0$ otherwise. Then we compute the Average Precision as  $\text{AP}(v) = \frac{\sum_{i} \text{Precision@i}(v)\cdot \mathds{1}(v, v_i)}{\sum_{i}\mathds{1}(v, v_i)}$.  Finally, MAP is given as:
\begin{equation*}
\text{MAP} = \frac{\sum_{v \in Q} \text{AP}(v)}{|Q|}
\end{equation*}

We compare our result with \nodevec only as it was previously shown to be better than other baselines \cite{grover2016node2vec}.

\par \noindent 
\textbf{Results.} See  Table \ref{tab:linkPredict}. Firstly, note that \alg outperforms \nodevec as $P$ varies from 10 to 30 in all the datasets.  We also notice that  \alg DM performs surprisingly worse than \nodevec and \alg DBON on \facebook. The reason for its poor performance in \facebook is that the network is dense with average clustering co-efficient of 0.6 and effective radius of 4 for 90\% of the nodes. Recall that the \alg DM optimization relies on finding the embedding of the nodes as well, which will not be discriminative for dense networks. In contrast, \alg DBON learns the features of subgraps directly, without relying on node embeddings, and hence it performs very well on large dense networks including \facebook. Finally we see that \nodevec consistently improves as $P$ increases, while both versions of \alg either deteriorate or stagnate. We discuss this more in Section~\ref{sec:conclusions}.

\begin{table*}[htb]
\centering
\caption{\textbf{Mean Average Precision for the link prediction task. P is the percentage of edge removed from the network and S stands for \alg. Winners have been bolded for each dataset. Either \alg DM or \alg DBON outperform \nodevec across all the datasets.}}

\label{tab:linkPredict}
\begin{tabular}{|c ? c|c|c ? c|c|c ? c|c|c? c|c|c?}
\hline &\multicolumn{3}{ c?}{\workplace} &\multicolumn{3}{c?}{\highschool} & \multicolumn{3}{c?}{\facebook} & \multicolumn{3}{c?}{\arxiv} \\ \hline

 $P$  & \nodevec & S DBON & S DM & \nodevec & S DBON & S DM &\nodevec & S DBON &S DM & \nodevec & S DBON & S DM \\ \hline 
 10&  0.25& \textbf{0.37} & 0.33 & 0.39 & 0.42 & \textbf{0.52} & 0.50 & \textbf{0.77} & 0.29& 0.12 & 0.24 & \textbf{0.31}\\ \hline
 20&  0.36& 0.28 & \textbf{0.42} & 0.41 & \textbf{0.52} & 0.26 & 0.68 & \textbf{0.84} & 0.34& 0.21 & \textbf{0.31} & 0.28 \\ \hline
 30&  0.39& 0.28 & \textbf{0.40} & 0.50 & 0.45 & \textbf{0.57} & 0.72 & \textbf{0.83} & 0.35& 0.26 & 0.37 & \textbf{0.44}\\ \hline

\end{tabular}
\end{table*}

\subsection{Parameter Sensitivity}

\begin{figure}[ht]
\begin{center}
  \begin{tabular}{c|c}
  \includegraphics[width=0.2\textwidth]{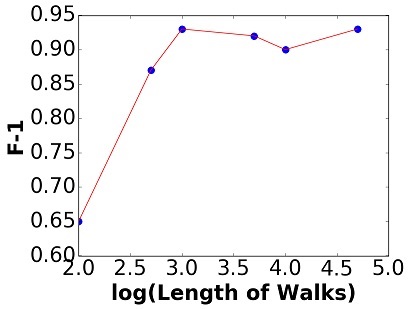} &
      \includegraphics[width=0.2\textwidth]{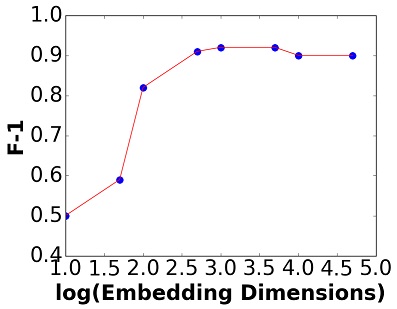}  \\
      (a) Walk length  & (b) Dimension of Vectors 
    \end{tabular}
    \vspace{-0.15in}
\caption{F-1 score on \polblogs for various values of walk length and dimension of embeddings.}

\label{sensitivity}

\end{center}
\end{figure}

Here we discuss the parameter sensitivity of \alg. We show how the F-1 score for community detection task on \polblogs dataset changes when we change the two parameters of \alg: (i) length of the random walk and (ii) dimension of the embedding. As shown in Figure \ref{sensitivity} (a), the F-1 score is 0.85 even when we do random walks of length 500. For the higher length,  the F-1 score remains constant. 

Similarly, to see how the results of the community detection task changes with the size of the embedding, we run the community detection task on \polblogs with varying embedding dimension. See  Figure  \ref{sensitivity} (b). The F-1 score saturates when the dimension of vector is greater than 100.

\subsection{Scalability}
\begin{figure}[ht]
\begin{center}
  \begin{tabular}{c|c}
  \includegraphics[width=0.2\textwidth]{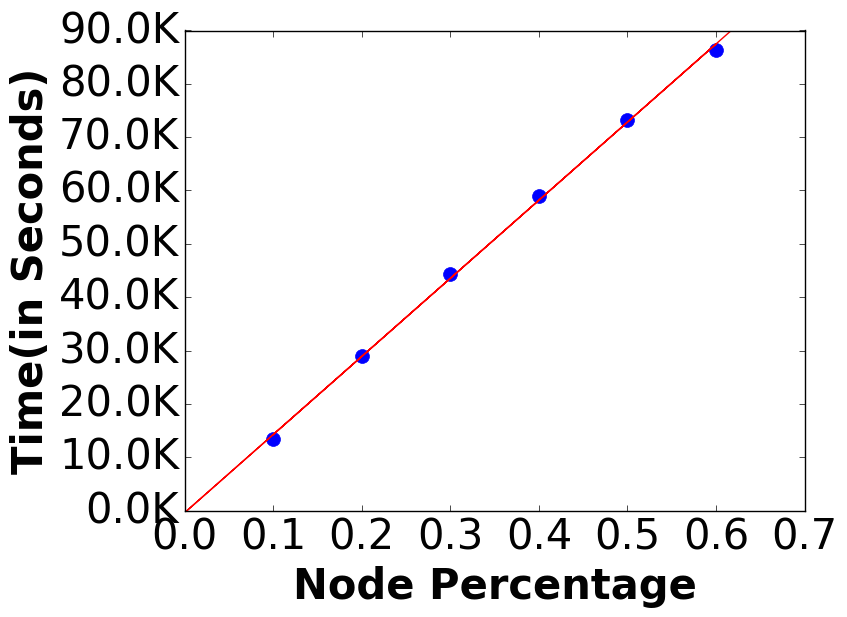} &
      \includegraphics[width=0.2\textwidth]{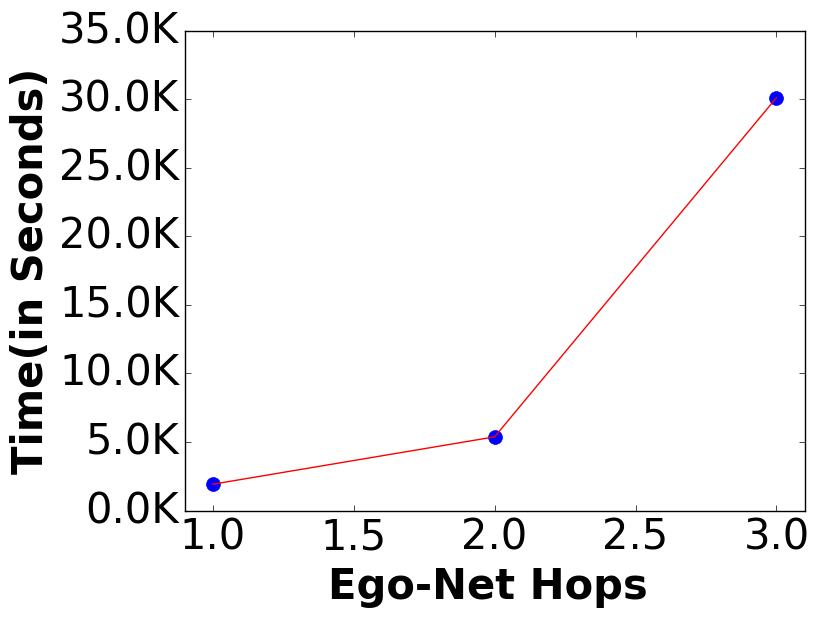}  \\
      (a) No of Subgraphs  & (b) Size of Subgraphs 
    \end{tabular}
    \vspace{-0.15in}
\caption{Scalability w.r.t. number of subgraphs on \youtube and w.r.t size of subgraphs on \arxiv datasets. }
\label{scalability}
\end{center}
\end{figure}

Here we show the scalability of \alg with respect to the number and the size of subgraphs. We extract connected subgraphs of \youtube dataset of induced by varying percentage of nodes. We then run \alg on the set of ego-nets in each resulting network. As shown in Figure \ref{scalability} (a), \alg is linear w.r.t number of subgraphs. 
In Figure \ref{scalability} (b), we run \alg on 1 to 3 hops ego-nets of \arxiv dataset. We see a significant jump in the running time when the hop increases from 2 to 3. This is due to the fact that as the hop of ego-net increases, the size of the subgraph increases exponentially due to the low diameter of real world networks.	

\subsection{Case Studies}
We perform case-studies on \memetracker\footnote{{snap.stanford.edu}}  and \dblp to investigate if our embeddings are interpretable. \memetracker consists of a series of cascades caused by memes spreading on the network of linked web pages. Each meme-cascade induces a subgraph in the underlying network. We first embed these subgraphs in a continuous vector space by leveraging \alg. We then cluster these vectors to explore what kind of meme cascade-graphs are grouped together, what characteristics of memes determine their similarity and distance to each other and so on. For this case-study, we pick the top 1000 memes by volume in the data. And we cluster them into 10 clusters using K-Means.

 We find coherent clusters which are meaningful groupings of memes based on topics. For example we find cluster of memes related to different topics such as  entertainment, politics, religion, technology and so on. Visualization of these clusters is presented in Figure \ref{fig:memetrackerCase}. In the entertainment cluster, we find memes which are names of popular songs and movies such as ``sweet home alabama'',``somewhere over the rainbow'', ``Madagascar 2'' and so on. Similarly, we also find a cluster of religious memes. These memes are quotes from the Bible. We also find memes related to politics and religion in the same cluster such as ``separation of church and state'''. In politics cluster, we find popular quotes from the 2008 presidential election season e.g. Barack Obama's popular slogan ``yes we can'' along with his controversial quotes like ``you can put lipstick on a pig'' in the cluster. We also find Sarah Palin's quote like ``the chant is drill baby drill''. Similarly, we also find a cluster of technology/video games related memes. 
 
 Interestingly, we find that all the memes in Spanish language were clustered together. This indicates that memes in different language travel though separate websites, which matches with the reality as most webpages use one primary language. We also noticed that some of the clusters did not belong to any particular topic. Upon closer examination we found out that these clusters contained memes which were covered by general news website such as msnbc.com, yahoo.com, news.google.com and local news websites such as philly.com from Philadelphia and breakingnews.ie from Ireland.

\begin{figure*}[htb]
\begin{center}

\begin{tabularx}{\textwidth}{>{\centering\arraybackslash}X>{\centering\arraybackslash}X>{\centering\arraybackslash}X}
		\includegraphics[width=0.4\textwidth]{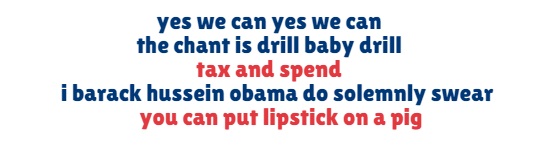} &
       \includegraphics[width=0.4\textwidth]{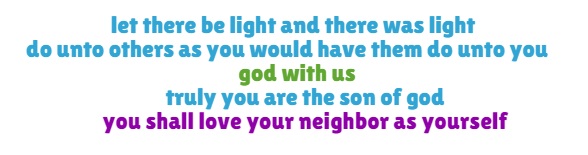} &
       \includegraphics[width=0.4\textwidth]{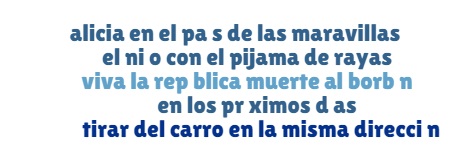} \\
      (a) Politics Cluster  & (b) Religion Cluster & (c) Spanish Cluster 
    \end{tabularx}
\end{center}

\begin{center}

\begin{tabularx}{\textwidth}{>{\centering\arraybackslash}X>{\centering\arraybackslash}X}
       \includegraphics[width=0.4\textwidth]{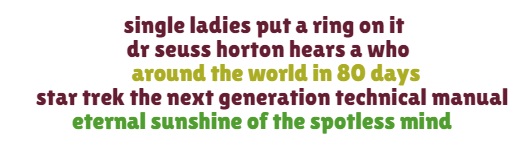} &
       \includegraphics[width=0.4\textwidth]{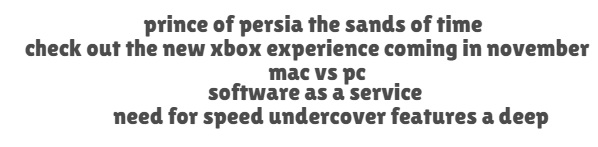} \\
      (d) Entertainment Cluster & (e) Technology Cluster 
    \end{tabularx}
\end{center}
	\vspace{-0.1in}
\caption{\textbf{Different Clusters of Memes for the \memetracker dataset.}}
\label{fig:memetrackerCase}
	\vspace{-0.1in}
\end{figure*}

For \dblp, we follow the methodology in~\cite{lappas2010finding}, and extract subgraphs of the coauthorship network based on the keywords contained in the title of the papers. 

We include keywords such as `classification', `clustering', `xml', and so on. Once we extract the subgraphs, we run \alg to learn embedding of these subgraphs. We then project the embeddings down to 2-dimensions using t-SNE~\cite{maaten2008visualizing}.

See Figure~\ref{fig:casedblp}. We see some meaningful groupings in the plot. We see that the keyword related to each other such as  `graphs', `pagerank', `crawling',  and `clustering' appear together. The classification related keywords such as `boosting', `svm', and `classification' are grouped together. We also see that `streams' and `wavelets' are close to each other. These meaningful groups of keywords highlight the fact that \alg results in meaningful embeddings.

\begin{figure}[htb]
	\includegraphics[width=0.5\textwidth]{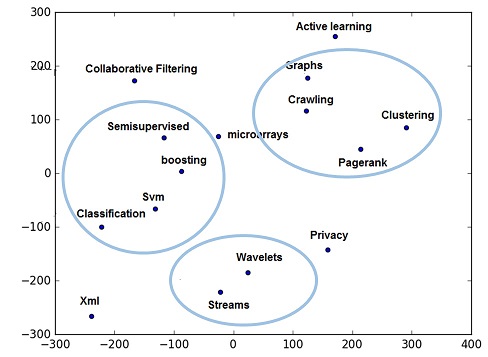}
		\vspace{-0.2in}
	\caption{2D projection of feature vectors learnt by \alg of subgraphs of \dblp induced by different keywords. }
	\label{fig:casedblp}
	\vspace{-0.2in}
\end{figure} 

\section{Related Work}
\label{sec:related}
\noindent \textbf{Network Embedding.}
The network embedding problem has been well studied. Most of work seeks to generate  low dimensional feature representation of nodes. Early work includes Laplacian Eigenmap~\cite{belkin2001laplacian}, IsoMap~\cite{tenenbaum2000global}, locally linear embedding~\cite{roweis2000nonlinear}, and spectral techniques~\cite{bach2003learning, chung1997spectral}.
Recently, several deep learning based network embeddings algorithms were proposed to learn feature representations of nodes~\cite{perozzi2014deepwalk, wang2016structural, tang2015line, grover2016node2vec}.
 Perozzi et. al \cite{perozzi2014deepwalk} proposed DeepWalk, which extends skip-Gram model~\cite{mikolov2013distributed} to networks and learns feature representation based on contexts generated by random walks. 
 Grover et. al. proposed a more general method, Node2Vec \cite{grover2016node2vec}, which generalizes random walks to generate various contexts. SDNE \cite{wang2016structural} and LINE \cite{tang2015line} learn feature representation of nodes while preserving first and second order proximity. However, all of them learn low dimensional feature vector of nodes, while our goal is to embed subgraphs.

 The most similar network embedding literature includes  \cite{riesen2010graph, yanardag2015deep, narayanan2016subgraph2vec}.
 Risen and Bunke propose to learn vector representations of graphs based on edit distance to a set of pre-defined prototype graphs~\cite{riesen2010graph}. Yanardag et. al. \cite{yanardag2015deep}  and  Narayanan et al. \cite{narayanan2016subgraph2vec} learn vector representation of the subgraphs using the Word2Vec~\cite{mikolov2013distributed} by 
  generating ''corpus'' of subgraphs  where each subgraph is treated as a word. 
 The above work focuses on some specific subgraphs like graphlets and rooted subgraphs. None of them embed subgraphs with arbitrary structure. In addition, we interpret subgraphs as paragraphs, and leverage the PV-DBOW and PV-DM model \cite{le2014distributed}.

\noindent \textbf{Other Subgraph Problems.}  
There has been a lot of work on subgraph related problems. For example, the subgraph discovery problems have been studies extensively. Finding the largest clique is a well-known NP-complete problem~\cite{karp1972reducibility}, which  is also hard to approximate~\cite{hstad1996clique}.  Lee et al. surveyed dense subgraph discovery algorithms for several subgraphs including clique, K-core, K-club, etc~\cite{lee2010survey}. Akoglu et al. extended the subgraph discovery problem to attributed graphs~\cite{akoglu2012pics}.
 Perozzi et al. studied the attributed graph anomaly detection by exploring the neighborhood subgraph of a nodes~\cite{perozzi2016scalable}.
Different from the above works, we seek to find feature representations of subgraphs.

\section{Discussion}
\label{sec:conclusions}
\begin{table*}[htb]
\centering
\caption{\textbf{Mean Average Precision for the link prediction task. $P$ is the percentage of edge removed and S stands for \alg.}}
\label{tab:linkPredicthigh}

\begin{tabular}{|c ? c|c|c ? c|c|c ? c|c|c? c|c|c?}
\hline &\multicolumn{3}{ c?}{\workplace} &\multicolumn{3}{c?}{\highschool} & \multicolumn{3}{c?}{\facebook} & \multicolumn{3}{c?}{\arxiv} \\ \hline

 $P$  & \nodevec & S DBON & S DM & \nodevec & S DBON & S DM &\nodevec & S DBON &S DM & \nodevec & S DBON & S DM \\ \hline 
  
 40&  0.45& 0.32 & 0.35 & 0.60 & 0.47 & 0.56 & 0.75 & 0.78 & 0.22& 0.30 & 0.39 & 0.33 \\ \hline
 50&  0.48 & 0.31 & 0.33 & 0.57 & 0.42 & 0.49 & 0.78 & 0.75 & 0.12& 0.33 & 0.26 & 0.34 \\ \hline
 60&  0.50& 0.33 & 0.32 & 0.60 & 0.40 & 0.43 & 0.79 & 0.53 & 0.1& 0.34 & 0.29 & 0.29\\ \hline
\end{tabular}
\end{table*}

We have shown that \alg gives meaningful interpretable embeddings of arbitrary subgraphs. We have also shown via our experiments that \alg outperforms traditional algorithms as well as node-level embedding algorithms for extracting communities from networks, especially in challenging dense graphs. Similarly for link prediction, we also showed that embedding neighborhoods is better for finding correct links.

So for which tasks will \alg not be ideal? For link prediction, as previously mentioned in Section \ref{sec:experiments}, the performance of \alg deteriorates when higher percentages of edges are removed from the network. The results for higher percentages, $P$ = 40 to 60, is presented in Table \ref{tab:linkPredicthigh}. The result shows that \nodevec outperforms \alg in such cases, despite performing poorly for lower values of $P$. 
This happens because, as $P$ increases, the density of the network decreases and results in lesser overlaps in the neighborhoods of nearby nodes. Hence \alg which preserves the local proximity of subgraphs, does not embed such subgraphs very close to each other, resulting in poorer prediction performance.
 
We believe, in such situations, perhaps using other proximity measures between subgraphs is more meaningful to preserve during the embedding process than only local proximity. 

One such way can be using `positional promixity', where two subgraphs are proximal based on their \emph{position} in the network. For example, in Figure \ref{fig:toy}, subgraphs induced by nodes  $\{c, d, e\}$ and $\{g, h, j\}$ are similar to each other as the member nodes in these two subgraphs have similar roles. Nodes $e$ and $h$ both connect to central node $f$ and nodes $d$ and $g$ both have degree two. Using just local proximity, these subgraphs are not similar. 
\par \noindent 
\textbf{Positional Proximity:} If we are given two subgraphs $g_i(v_i, e_1)$ and $g_j(v_j, e_j)$, then the \textbf{positional proximity} between $g_i$ and $g_j$ is determined by similarity of position of nodes in $g_i$ and $g_j$.

Similarly, another way can be using similarity based on \emph{structure} of subgraphs. For example, in Figure \ref{fig:toy}, subgraphs induced by nodes $\{a, b, c, e \}$ and $\{h, i, j, k\}$ are similar to each other as both of them are cliques of size four.

\par \noindent \textbf{Structural Proximity:} If we are given two subgraphs $g_i(v_i, e_1)$ and $g_j(v_j, e_j)$, then the \textbf{structural proximity} between $g_i$ and $g_j$ is determined by the structural properties of $g_i$ and $g_j$. 

For link prediction in very sparse networks, Positional Proximity might give more useful embeddings than Local Proximity. 
We leave the task of embedding subgraphs based on Structural and Positional proximities (or using a combination with Local proximity) and leveraging them for graph mining as future work.

\section{Conclusion}
We have presented \alg, a scalable feature learning framework for a set of subgraphs such that the local proximity between them are preserved. In contrast most prior work focused on finding node-level embeddings. 
We give a theoretical justification and showed that the embeddings generated by \alg can be leveraged in downstream applications such as community detection and link prediction. We also performed case-studies on two real networks to validate the usefulness of the subgraph features generated by \alg.

\bibliographystyle{ACM-Reference-Format}
\bibliography{all,references,related}


\end{document}

%% file: dfn.tex
\newcommand{\beq}{\begin{equation}}
\newcommand{\eeq}{\end{equation}}
\newcommand{\bit}{\begin{itemize}}
\newcommand{\eit}{\end{itemize}}

\newcommand{\hide}[1]{}
\newtheorem{problem}{Problem}